\documentclass[11pt,letterpaper]{article}

\usepackage[margin=1in]{geometry}
\usepackage{microtype}
\usepackage{graphicx}
\usepackage{caption}
\usepackage{booktabs} 
\usepackage[normalem]{ulem}
\usepackage[symbol]{footmisc}
\usepackage[hidelinks]{hyperref}

\usepackage{natbib}

\usepackage{amsmath}
\usepackage{amssymb}
\usepackage{mathtools}
\usepackage{amsthm}
\usepackage{subcaption}

\usepackage{xspace}

\usepackage{thmtools} 
\usepackage{thm-restate}

\usepackage[capitalize,noabbrev]{cleveref}

\usepackage{algorithm}
\usepackage[noend]{algorithmic}

\theoremstyle{plain}
\newtheorem{theorem}{Theorem}[section]

\newtheorem{lemma}[theorem]{Lemma}

\theoremstyle{definition}

\theoremstyle{remark}

\newtheorem{claim}[theorem]{Claim}
\newtheorem{observation}[theorem]{Observation}

\Crefname{claim}{Claim}{Claims}

\newcommand{\OPT}{\operatorname{OPT}}

\newcommand{\bbRp}{{\mathbb{R}_{\ge 0}}}

\newcommand{\rb}[1]{\left( #1 \right)} 

\newcommand{\Init}{\textsc{Initialization}\xspace}
\newcommand{\swapping}{\textsc{Swapping}\xspace}
\newcommand{\Insertion}{\textsc{Insertion}\xspace}

\newcommand{\Deletion}{\textsc{Deletion}\xspace}
\newcommand{\LevelConstruct}{\textsc{Level-Construct}\xspace}

\newcommand{\lazygreedy}{\textsc{Lazy-Greedy}\xspace}

\newcommand{\fullygreedy}{\textsc{Dynamic-Greedy}\xspace}
\newcommand{\fullyswapping}{\textsc{Dynamic-Swapping}\xspace}

\newcommand{\thresholdswapping}{\textsc{Threshold-Swapping}\xspace}

\newcommand{\marginal}[2]{f\rb{#1 \mid #2}}

\newcommand{\ellstar}{\ell^{\star}}
\newcommand{\eps}{\varepsilon}

\newcommand{\cM}{\mathcal{M}}

\newcommand{\sol}[1]{}

\newcommand{\rank}{k}
\newcommand{\appendixref}{\if\fullversion1 \cref{appendix:extra-experiments}\xspace \else the Appendix\xspace \fi}

\DeclareMathOperator*{\argmin}{argmin}

\newcommand{\E}[1]{\mathbb{E}\left[#1\right]}

\newcommand{\shortN}{N^-_{\ell}}
\newcommand{\longN}{N^+_{\ell}}

\title{Fully Dynamic Submodular Maximization over Matroids}

\author{
	Paul D{\"u}tting\thanks{Google Research.}
	\and
	Federico Fusco\thanks{Department of Computer, Control, and Management 
Engineering, Sapienza University of Rome, 
Italy.}
\and
Silvio Lattanzi{$^*$}
\and
Ashkan Norouzi-Fard{$^*$}
\and
Morteza Zadimoghaddam{$^*$}}

\date{}

\begin{document}

\maketitle

\begin{abstract}
    Maximizing monotone submodular functions under a matroid constraint is a classic algorithmic problem with multiple applications in data mining and machine learning. We study this classic problem in the fully dynamic setting, where elements can be both inserted and deleted in real-time. 
    Our main result is a randomized algorithm that maintains an efficient data structure with an $\tilde{O}(k^2)$ amortized update time (in the number of additions and deletions) and yields a $4$-approximate solution, where $k$ is the rank of the matroid. 
\end{abstract}

\section{Introduction}

Thanks to the ubiquitous nature of ``diminishing returns'' functions, submodular maximization is a central problem in unsupervised learning with multiple applications in different fields, including video analysis~\citep{ZhengJCP14}, data summarization~\citep{LinB11,BairiIRB15}, sparse reconstruction~\citep{Bach10,DasK11}, and active learning~\citep{GolovinK11,AmanatidisFLLR22}.
    
    Given a submodular function $f$, a universe of elements $V$, and a family $\mathcal{F} \subseteq 2^V$ of subsets of $V$ the submodular maximization problem consists in finding a set $S \in \mathcal{F}$ that maximizes $f(S)$. A classic choice for $\mathcal{F}$ are the capacity constraints (a.k.a.~$k$-uniform matroid constraints) where every subset $S$ of cardinality at most $k$ is feasible. Another common restriction that generalizes capacity constraints and comes up in many real-world scenarios are matroid constraints.
    Submodular maximization under matroid constraints is NP-hard, although efficient approximation algorithms exist for this task in both the centralized and streaming setting \citep{fisher78-II,CalinescuCPV11,ChakrabartiK15,EneN19a}.
    
    One fundamental limitation of these algorithms is that they are not well-suited to handle highly dynamic datasets, where elements are added and deleted continuously. Many real-world applications exhibit such dynamic behaviour; for example, \citet{DeyJR12} crawled two snapshots of 1.4 million New York City Facebook users several months apart and reported that 52\% of the users changed their profile privacy settings during this period. Similarly, TikTok processes millions of video uploads and deletions each day, while also Snapchat processes millions of message uploads and deletions daily. In such settings, it is essential to quickly perform basic machine learning tasks, such as active learning or data summarization, so it is crucial to design \emph{fully dynamic} algorithms that can \emph{efficiently} process streams containing not only insertions but also an arbitrary number of deletions, with small processing time per update.
    
    For these reasons, many problems have been studied in the dynamic setting, even if it is notoriously difficult to obtain efficient algorithms in this model. For monotone submodular maximization with a cardinality constraint, a $(2+\eps)$-approximation algorithm with poly-logarithmic amortized update time (with respect to the length of the stream) was designed by \citet{LattanziMNTZ20}; subsequently, this result has been proved to be tight by \citet{ChenP22}. 
    In the case of submodular maximization with matroid constraints, algorithms have been proposed only for specialized dynamic settings, namely sliding windows~\citep{ChenNZ16,EpastoLVZ2017} and deletion robustness~\citep{DuettingFLNZ22,MirzasoleimanK017,ZhangTG22}.
    
    \paragraph{Our contribution.} In this paper we propose the first fully dynamic algorithm for submodular maximization under a matroid constraint with amortized running time that is sublinear in the length of the stream. Our randomized algorithm processes a stream of arbitrarily interleaved insertions and deletions with an (expected) amortized time per update that is $\tilde O(k^2)$\footnote{In this work, $\tilde{O}$ hides factors poly-logarithmic in $n$ (the number of elements in the stream) and $k$ (the rank of the matroid).}. Crucially, it also continuously maintains a solution whose value is (deterministically), after each update, at least $\frac14$ of the optimum on the available elements. 
    
    \paragraph{Technical challenges.} While many known algorithms handle insertions-only streams, it is challenging to efficiently handle deletions: removing one element from a candidate solution may make necessary to recompute a new solution from scratch using {\em all} the elements arrived in previous insertions. This is the reason why well-known techniques for the centralized or streaming framework cannot be applied directly in the dynamic setting without suffering a linear amortized update time $\Omega(n)$ (see \cref{app:benchmark} for further discussion). The fully-dynamic algorithm for cardinality constraint~\citep{LattanziMNTZ20} addresses this phenomenon via a two-dimensional bucketing data structure that allows to efficiently recover elements with large enough contribution to the current solution (and can be used to quickly recompose a good solution after a deletion). Unfortunately, that approach crucially depends on the nature of the constraint and does not extend to more structured constraints as matroids. The key difficulty is that when an element of an independent set in a matroid gets deleted, only a subset of the elements can replace it, according to the matroid constraint. This is a crucial difference with cardinality constraints, where all elements are interchangeable. 
    
    \paragraph{Our techniques.} In this paper, we also design and analyze a data structure that is organized in levels, each one providing robustness at different scales. In addition, we carefully design an update rule that simulates in real-time the behavior of the classic \swapping algorithm for submodular maximization under matroid constraint \citep{ChakrabartiK15}. A key insight of our approach is that one can reorder and delay the addition or swapping of the elements with lower robustness without losing the simplicity and effectiveness of the \swapping algorithm. Interestingly, our construction simplifies substantially that of \citet{LattanziMNTZ20} as it removes one of the two dimensions of the dynamic data structure. Finally, we highlight a speed-up to the swapping algorithm that reduces the number of matroid independence queries by a factor $(k/ \log k)$. This result may be of independent interest.

    \paragraph{Additional related works.} Independently and in parallel from the work of  \citet{LattanziMNTZ20}, \citet{Monemizadeh20} achieved the same approximation guarantee with $\tilde O(k^2)$ amortized update time were $k$ is the cardinality constraint. An area of research that is very close to the fully dynamic setting is robust submodular optimization  \citep{orlin2018robust, bogunovic2017robust, MirzasoleimanK017, MitrovicBNTC17, KazemiZK18, AvdiukhinMYZ19, Zhang22}. In this setting, the goal is to select a summary of the whole dataset that is robust to $d$ adversarial deletions; crucially the number $d$ of deletions is known to the algorithm and typically all the deletions happen after the insertion in the stream. The results in this line of research do not apply to our dynamic setting where the number of deletions is arbitrary and deletions are interleaved with insertions.
    
\section{Preliminaries} 
\label{sec:prel}
    We consider a set function $f: 2^V \to \bbRp$ on a (potentially large) ground set $V$. Given two sets $X, Y \subseteq V$, the \emph{marginal gain} of $X$ with respect to $Y$, $\marginal{X}{Y}$, quantifies the change in value of adding $X$ to $Y$ and is defined as $  \marginal{X}{Y} = f(X \cup Y) -  f(Y).$

    When $X$ consists of a singleton $x$, we use the shorthand $f(x\mid Y)$ instead of $f(\{x\}\mid Y)$. Function $f$ is called \emph{monotone} if $\marginal{e}{X}  \geq 0$ for each set $X \subseteq V$ and element $e \in V$, and \emph{submodular} if for any two sets $X \subseteq Y \subseteq V$ and any element $e \in V \setminus Y$ we have 
    \[
        \marginal{e}{X} \ge \marginal{e}{Y}.
    \]
    Throughout the paper, we assume that $f$ is monotone and that it is \emph{normalized}, i.e., $f(\emptyset) = 0$. We model access to the submodular function $f$ via a value oracle that computes $f(S)$ for given $S \subseteq V$.

    \paragraph{Submodularity under a matroid constraint.}
    A non-empty family of sets $\cM \subseteq 2^V$ is called a \emph{matroid} if it satisfies the following properties:
    \begin{itemize}
        \item{\emph{Downward-closure}} if $A \subseteq B$ and $B \in \cM$, then $A \in \cM$
        \item{\emph{Augmentation}} if $A, B \in \cM$ with $|A| < |B|$, then there exists $e \in B$ such that $A + e \in \cM$.
    \end{itemize}
    For the sake of brevity, in this paper we slightly abuse the notation and for a set $X$ and an element $e$, use $X+e$ to denote $X \cup \{e\}$ and $X - e$ for $X \setminus \{e\}$.
    We call a set $A \subseteq 2^V$ \emph{independent}, if $A \in \cM$, and \emph{dependent} otherwise. An independent set that is maximal with respect to inclusion is called a {\em base}; all the bases of a matroid share the same cardinality $k$, which is referred to as the {\em rank} of the matroid. The problem of maximizing a function $f$ under a \emph{matroid constraint} $\cM$ is defined as selecting a set $S \subseteq V$ with $S \in \cM$ that maximizes $f(S)$. Similarly to what is done for the submodular function, we assume access to an independence oracle that takes in input $S\subseteq V$ and outputs whether $S$ is independent with respect to the matroid or not.

    \paragraph{Fully dynamic model.}
    Consider a stream of exactly $n$ insertion and $n$ deletion operations chosen by an oblivious adversary. Denote by $V_i$ the set of all elements inserted and not deleted up to the $i$-th operation. Let $O_i$ be an optimum solution for $V_i$ and denote $\OPT_i = f(O_i)$. Our goal is to design a dynamic data structure with two key properties.
    On the one hand, we want the data structure to maintain, at the end of each operation $i$, a good feasible solution $S^i \subseteq V_i$. In particular, we say that an algorithm is an $\alpha$-approximation of the best (dynamic) solution if $\OPT_i \le \alpha f(S^i)$, for all $i = 1,\dots, 2n$.
    On the other hand, we are interested in updating our data structure efficiently. We measure efficiency in terms of the amortized running time, i.e., the average per-operation computation: we say that an algorithm has amortized running time $t$ if its expected total running time to process any stream of $2n$ insertions and deletions is at most $2nt$.\footnote{We are interested in the asymptotic behaviour of the amortized running time, therefore we can safely assume that the sequence contains exactly $n$ deletions.}
    Throughout this paper, we refer to running time as the total number of submodular function evaluations (value oracle) and independent set evaluations with respect to the matroid (independence oracle). This is a standard practice in submodular optimization as these two oracles typically dominates the running time of optimization algorithms. 
    
    \begin{algorithm}[t]
	\caption{\swapping} \label{alg:swapping}
	\begin{algorithmic}[1]
		\STATE \textbf{Environment:} stream $\pi$, function $f$ and matroid  $\cM$
		\STATE $S \gets \emptyset$, $S' \gets \emptyset$
		\FOR{each new arriving element $e$ from $\pi$}
		    \STATE $w(e) \gets f(e\mid S')$
		    \IF{$S + e \in \cM$}
		        \STATE $S \gets S + e$, $S' \gets S' + e$
		    \ELSE 
                \STATE $s_e \gets \argmin\{w(y) \mid y \in S, \ e + S - y \in \cM\}$ \label{line:swap-consistent}
                \IF{$2 w(s_e) < w(e)$}
                    \STATE $S \gets S - s_e + e$, $S' \gets S' + e$
                \ENDIF
		    \ENDIF
		\ENDFOR
		\STATE \textbf{Return} $S$
	\end{algorithmic}
    \end{algorithm}
    \paragraph{Insertion-only streams.} The fully dynamic model can be considered --- to some extent --- a generalization of the insertion-only streaming model. There, an arbitrary sequence of sole insertions is passed to the algorithm that is tasked with retaining a good solution (with respect to the offline optimum), while using only little ``online'' memory. A key ingredient in our analysis is the \swapping algorithm by \citet{ChakrabartiK15} that is a simple yet powerful routine for submodular maximization with matroid constraint in the streaming setting. \swapping maintains a feasible solution and, for each new arriving element, it adds it to the solution if either it does not violate the matroid constraint or it is possible to swap it with some low-value element\footnote{With ties in line \ref{line:swap-consistent} solved in any consistent way.}. We use a slightly modified version of the original algorithm (see pseudocode for details); namely, the weight of a new element is computed as its marginal value with respect to the set $S'$ of all the elements that {\em at some point} were in the solution. We refer to \Cref{app:swapping} for a formal proof of the fact that our modified version of \swapping still retains the approximation guarantees we want: 
    \begin{theorem}
    \label{thm:swapping}
        For any (possibly adaptive) stream of elements in $V$, \swapping outputs a deterministic $4$-approximation to the best (offline) independent set in $V.$ 
    \end{theorem}
    
\section{The Algorithm}

    In the main body of the paper we present and analyze a simplified version of our data structure whose amortized running time depends poly-logarithmically on a parameter $\Delta$ of the function $f$:
    \[
        \Delta = \frac{\max_{x \in V} f(x)}{ \min_{T \subseteq V, x \notin T_0} f(x\mid T)},
    \]
    where with $T_0$ we denote the set of all the elements with $0$ marginal contribution with respect to $T$. In \Cref{app:Delta}, we show how to replace this dependence in $\Delta$ with a $O(k/\eps)$ term for any chosen precision parameter $\eps$ that influences the approximation factor in an additive way. To further simplify the presentation, we also assume without loss of generality
    that our algorithm knows the number $n$ of insertions and deletions in advance, and that $n$ is a power of $2$. We show in \cref{sec:all-together} a simple way to avoid this assumption without affecting the approximation guarantee.
    
    We are ready to introduce our algorithm. At a high level, it carefully maintains the stream of elements in a data structure characterized by a small amortized update time and that mimics the behavior of \swapping, at each insertion or deletion operation. 
    Our data structure contains $L +1$ levels, with $L = \log n$. Each one of these levels is characterized by four sets of elements: a partial solution $S_{\ell}$, an auxiliary set $S'_{\ell}$ that contains $S_{\ell}$ and some elements that used to belong to $S_{\ell}$ but were later swapped out from it, a set $A_{\ell}$ of candidate elements, that meet certain criteria and are considered {\em good} addition to the solution, and a buffer $B_{\ell}$ of still not processed elements. Moreover, the invariants that $|A_{\ell}|$ and $|B_{\ell}|$ are smaller than $n/2^\ell$ are enforced. We claim that the solution of the last level, i.e., $S_{L}$ (that plays the role of $S^i$), is consistently a constant factor approximation of $\OPT_i$, at the end of each operation $i$. We describe how the data structure is maintained; this clarifies the details of our approach. 

    \paragraph{Initialization.} At the beginning of the stream, the routine \Init is called. It takes in input $n$ and initializes $ \Theta(\log n)$ empty sets; the ones described above: $S_{\ell}, S'_{\ell}, A_{\ell}$ and $B_{\ell}$ for all $\ell = 0, 1, \dots, L$.
    \begin{algorithm}[H]
	\caption{\Init} \label{alg:initialization}
	\begin{algorithmic}[1]
	    \STATE \textbf{Input:} $n$
		\STATE $L \leftarrow \log n$
		\STATE Initialize empty sets $A_{\ell}, S_{\ell}, S'_{\ell}, B_{\ell}$ $\quad \forall \, 0 \leq \ell \leq L$ 
	\end{algorithmic}
    \end{algorithm}
    \paragraph{Handling insertions.} When a new element $e$ is inserted, it gets immediately added to all the buffers (line \ref{line:insertion-add-e} of \Insertion). This addition induces the call of another routine, \LevelConstruct, on the level $\ell$ with smallest index such that the buffer $B_{\ell}$ exceeds a certain cardinality (namely when $|B_{\ell}| \ge n/2^\ell$, see line \ref{line:insertion-check} of \Insertion). Such level always exists by our choice of $L$.   
    \begin{algorithm}[H]
    	\caption{$\Insertion(e)$} \label{alg:insertion}
    	\begin{algorithmic}[1]
    		\STATE $B_\ell \gets B_\ell + e \quad \forall \, 0 \leq \ell \leq L$ \label{line:insertion-add-e}
    		\IF{there exists an index $\ell$ such that $|B_\ell| \ge \frac{n}{2^\ell}$}\label{line:insertion-check}
    		    \STATE Let $\ellstar$ be such $\ell$ with lowest value 
        		\STATE Call $\LevelConstruct(\ellstar)$ \label{line:insertion-invokes-level-construct}
    		\ENDIF
    	\end{algorithmic}
    \end{algorithm}
    \paragraph{Handling deletions.} When an element $e$ is deleted from the stream, then the data structure is updated according to \Deletion. Element $e$ is removed from all the candidate elements sets $A_{\ell}$ and buffers $B_{\ell}$ (lines \ref{line:deletion_candidate} and \ref{line:deletion_buffers}) and causes a call of \LevelConstruct on the smallest-index level such that $e \in S_{\ell}$ (line \ref{line:deletion-level}). While \Insertion always induces a call of \LevelConstruct, \Deletion only causes it if the deleted element belongs to some partial solution $S_{\ell}$. 
    \begin{algorithm}[H]
    	\caption{$\Deletion(e)$} \label{alg:deletion}
    	\begin{algorithmic}[1]
    		\STATE $A_{\ell} \gets A_{\ell} - e \quad \forall \, 0 \leq \ell \leq L$
    		\label{line:deletion_candidate}
    		\STATE $B_\ell \gets B_\ell - e \quad \forall \, 0 \leq \ell \leq L$
    		\label{line:deletion_buffers}
    		\IF {$e \in S_{\ell}$ for some $\ell$}
    		    \STATE Let $\ell$ be the smallest index such that $e \in S_{\ell}$
    		    \STATE Call $\LevelConstruct(\ell)$ 
    		    \label{line:deletion-level}
    		\ENDIF
    	\end{algorithmic}
    \end{algorithm}
    \paragraph{\LevelConstruct.} We now describe the main routine of our data structure: \LevelConstruct. A call to this routine at level $\ell$ triggers some operations relevant to sets at level $\ell$, and it then recursively runs \LevelConstruct at level $\ell + 1$. Therefore $\LevelConstruct(\ell)$ is essentially responsible for reprocessing the whole data structure at all levels $\ell, \ell+1, \cdots, L$. 
    When it is called on some level $\ell$, all the sets associated to that level ($S_{\ell}, S'_{\ell}, A_{\ell}$ and $B_{\ell}$) are reinitialized: the candidate elements set $A_{\ell}$ is initialized with the elements in $A_{\ell-1}$ and $B_{\ell-1}$ (line \ref{line:level-first-set-a}), the buffer $B_{\ell}$ is erased (line \ref{line:level-b-empty}), while $S_{\ell}$ and $S'_{\ell}$ are copied from the previous level (lines \ref{line:level-s_l} and \ref{line:level-s'_l}). Then, the following iterative procedure is repeated, until the cardinality of $A_{\ell}$ becomes smaller or equal to $n/2^\ell$: first, all the elements in $A_{\ell}$ that would not be added to $S_{\ell}$ by \swapping are filtered out (lines \ref{line:E_l} to \ref{line:level-set-a-loop}), then, if the cardinality of $A_\ell$ is still large enough (i.e., $|A_\ell| \ge n/2^\ell$, see line \ref{line:level-second-check}) an element $e$ from it is drawn uniformly at random and is added to the solution and to $S'_{\ell}$ (lines \ref{line:add_S_l} and \ref{line:add_S'_l}); note that if $S_{\ell} + e \notin \cM$, then $e$ needs to be swapped with some element $s_e$ in the solution (see line \ref{line:find-swap}). Two important implementation details are worth mentioning here: $(i)$ every time an element $e$ is added to a partial solution $S_{\ell}$, also the information about the weight $w(e)$ it has at the moment is stored in $S_{\ell}$; $(ii)$ the partial solutions $S_\ell$ are maintained sorted in increasing order of weight. Note that these two points do not entail any call of the value or independence oracles.

\begin{algorithm}[H]
	\caption{$\LevelConstruct(\ell)$} \label{alg:level-construct}
	\begin{algorithmic}[1]
		\STATE $A_{\ell} \gets A_{\ell-1} \cup B_{\ell-1}$  \label{line:level-first-set-a}
		\STATE $B_\ell \gets \emptyset$ \label{line:level-b-empty}
		\STATE $S_{\ell} \gets S_{\ell-1}$ \label{line:level-s_l}
		\STATE $S'_{\ell} \gets S'_{\ell-1}$ \label{line:level-s'_l}
        \REPEAT	\label{line:level-from}
        \FOR{any element $e \in A_\ell$}\label{line:begin-swap-loop}
            \STATE $w(e) \gets f(e \mid S'_{\ell})$ \label{line:level-weight}
            \STATE $s_e \gets \argmin\{w(y) \mid y \in S_{\ell} \wedge S_{\ell} - y + e \in \cM\}$\label{line:find-swap}
        \ENDFOR
        \STATE $E_{\ell} \gets \{ e \in A_{\ell} \mid S_{\ell} + e \in \cM\}$ \label{line:E_l}
		\STATE $F_{\ell} \gets \{ e \in A_\ell \setminus E_{\ell} \mid w(e) > 2 \, w(s_e) \}$ \label{line:F_l}
		\STATE $A_{\ell} \gets E_{\ell} \cup F_{\ell}$ \label{line:A_l_filter}
        \label{line:level-set-a-loop}%
        \IF {$|A_{\ell}| \ge \frac{n}{2^\ell}$} \label{line:level-second-check}
        	\STATE Pop $e$ from $A_{\ell}$ uniformly at random \label{line:level-sample}
            \STATE $S_{\ell} \gets S_{\ell} + e - s_e$ 
            \label{line:add_S_l}
            \STATE $S'_{\ell} \gets S'_{\ell} + e$\label{line:add_S'_l}
        \ENDIF
        \UNTIL {$|A_{\ell}| < \frac{n}{2^\ell}$ } \label{line:level-until}
        \STATE \textbf{if} $\ell<L$, call $\LevelConstruct(\ell+1)$. \label{line:level-next-level}
    \end{algorithmic}
\end{algorithm}

\section{Approximation Guarantee}\label{sec:appr}

    Fix any operation $i$, we want to show that the solution $S_L$ maintained in the data structure at the end of the computation relative to operation $i$ is a good approximation of the best independent set $O_i$ (of value $\OPT_i$) in $V_i$. To not overload the notation, we omit the index $i$ when it is clear from the context, so that $V$ stands for the elements that were inserted but not deleted from the stream up to operation $i$ (included) and $D$ stands for the set of elements deleted from the stream up to operation $i$ (included). Clearly the set of all the elements arrived up to operation $i$ is exactly $V \cup D$. We want to show that $f(S_{L})$ is a (deterministic) $4$-approximation of the independent set in $V$ with largest value. In the following, we actually we prove that $f(S_{L})$ is a $4$-approximation of something that is at least $\OPT$.
    
    Up to operation $i$ the content of the data structure has changed multiple times, but for the sake of the analysis it is enough to consider a subset of all these modifications. Formally, for each level $\ell$, denote with $i_{\ell}$ the last operation that triggered a call of \LevelConstruct$(\ell)$ before or at operation $i$. This may have happened either directly via an insertion or a deletion at level $\ell$ or indirectly because something was inserted or removed at some level with smaller index. Denote with $X_{\ell}$ the elements in $V_{i_{\ell}}$ that were added to $S_{\ell}$ during the computation and with $Y_{\ell}$ the elements that were filtered out due failure in the ``swapping'' test. Formally, $X_{\ell} = S_{\ell} \setminus S_{\ell-1}$ at the end of the computation relative to operation $i_{\ell}$, while $Y_{\ell}$ is made up of all the elements that were in $A_{\ell} = B_{\ell-1} \cup A_{\ell-1}$ at the beginning of that call of \LevelConstruct$(\ell)$, but that were neither passed to the following level nor added to the solution during the repeat loop of that call of \LevelConstruct. Note that, by definition of $i_{\ell}$, nothing happens in level $\ell$ between operations $i_{\ell}$ and $i$ besides possibly some basic addition induced by \Insertion (line \ref{line:insertion-add-e}) and deletions induced by \Deletion (lines \ref{line:deletion_candidate} and \ref{line:deletion_buffers}), thus sets $X_{\ell}$, $Y_{\ell}$ and $S_{\ell}$ do not change.
    We start proving that all the $X_{\ell}$ and $Y_{\ell}$ are pair-wise disjoint. 
    \begin{restatable}{lemma}{lemdisjoint}
        All the $2L+2$ sets $X_{\ell}$ and $Y_{\ell}$ are pair-wise disjoint. 
    \end{restatable}
    \begin{proof}
        By definitions, it is immediate to see that $X_{\ell} \cap Y_{\ell} = \emptyset$ for  each level $0 \leq \ell \leq L$. Now, consider any two levels $\ell$ and $r$, with $\ell < r$. Since level $\ell$ has a smaller index, the last operation $i_{\ell}$ during whose computation $\LevelConstruct(\ell)$ was called is precedent (or equal) to $i_r$, that is the equivalent operation for level $r$. This means that any element $e$ in $X_{\ell} \cup Y_{\ell}$ do not belong to $X_r \cup Y_r$: $e$ was not passed to level $\ell+1$ (and thus to any larger level, $r$ included) during the computation at operation $i_{\ell}$ and, by definition of $i_{\ell}$ it was not considered (i.e., it did not appear in any call of line \ref{line:level-first-set-a} of \LevelConstruct) in any level with larger index, for all the computations between operation $i_{\ell}$ and operation $i$ ($i_r$ included!).
    \end{proof}
    
    Denote with $X$ the (disjoint) union of all the $X_{\ell}$ and the $Y_{\ell}$. We prove that $X$ is a superset of $V$.
    
    \begin{restatable}{lemma}{lempartition}
    \label{lem:partition}
         $V$ is contained in $X$: for any $e \in V$ there exists (exactly) one level $\ell$ such that $e \in X_{\ell}$ or $e \in Y_{\ell}$.
    \end{restatable}
    \begin{proof}
        When a new element $e \in V$ is inserted in the stream, it is added to all buffers (line \ref{line:insertion-add-e} of \Insertion) and triggers a call of \LevelConstruct at some level $\ell^*$ (line \ref{line:insertion-invokes-level-construct} of \Insertion). Thus it is either added to the solution or filtered out at some level. However, this is not enough to prove the Lemma, as it is possible that that call of \LevelConstruct is not the last time that element $e$ is considered.
        
        To formally prove the Lemma we introduce two (moving) auxiliary levels $u_e$ and $d_e$ such that the following two invariants hold from the operation in which $e$ is added onwards (up to operation $i$, included):
        \begin{itemize}
            \item[$a)$] $e$ belongs to the buffer $B_{\ell}$ for all $\ell < d_e$
            \item[$b)$] $e$ belongs to  $A_{\ell}$ for all $ d_e \le \ell < u_e$ 
            \item[$c)$] $e$ belongs to either $X_{u_e}$ or $Y_{u_e}$, for some $u_e \ge d_e$. 
        \end{itemize} 
        Stated differently, we show that at the end of each operation $j$ that follows the insertion of $e$ (included) up to operation $i$, included, there exist two levels (possibly different for each operation) such that $a),$ $b)$ and $c)$ hold. For the sake of clarity, we omit the dependence of $u_e$ and $d_e$ from $j$.
        
        When element $e$ is inserted, it triggers \LevelConstruct at some level (we call $d_e$ such level and note that $a)$ is respected) and it is either added or filtered out at some other level (we call $u_e$ such level so note that also $b)$ and $c)$ are respected). By construction of \LevelConstruct, it holds that $u_e \ge d_e$. 
        
        Note that $e \in V$, thus $e$ is not deleted from the stream before operation $i$. So we only need to show that any further call of \LevelConstruct happening between the insertion of $e$ and operation $i$ (included) does not affect the invariants. We have three cases. For any \LevelConstruct that is called for some $\ell > u_e$, nothing changes, as levels $\ell \le u_e$ are not touched.
        If \LevelConstruct is called for some $\ell < d_e$, then element $e$ belongs to the buffer $B_{\ell-1}$ (by the invariant $a$) and it is then added to either $X_{u_e}$ or $Y_{u_e}$ for some (new) $u_e \ge \ell$. We rename $d_e$ the above $\ell$, and it is easy to verify that all the invariants still hold. Finally, if \LevelConstruct is called for some $d_e \le \ell < u_e$, then by invariant $b$, it holds that $e$ belongs to $A_{\ell-1}$, thus it will end up filtered out or added to the solution in some new $u_e \ge \ell$. In this case we do not change $d_e$, and it is easy to see that the invariants still hold.  So the three invariants hold during the execution of the entire stream. To conclude the proof we just note that the invariants imply that $e$ is only contained in either one $X_\ell$ or one $Y_\ell$ for some $\ell$.
    \end{proof}
    
    There is a natural notion of ordering on the elements of $X$, induced by the order in which they were considered by the algorithm, i.e. in which they were either added to the solution $S_{\ell}$ (line \ref{line:level-s_l} of \LevelConstruct) or filtered out by the recomputation of $E_{\ell}$ and $F_{\ell}$ (line \ref{line:level-set-a-loop} of \LevelConstruct), with ties broken arbitrarily. Call $\pi$ this ordering. To have a better intuition of $\pi$, note that it can be split into contiguous intervals, the first corresponding to the elements considered in the first level $X_0 \cup Y_0$, then in the second $X_1 \cup Y_1$, and so on. In interval $\ell$, elements of $X_{\ell} \cup Y_{\ell}$ are ordered using the same order in which they have been added or filtered out in the last call of \LevelConstruct$(\ell)$. 
    
    The crucial observation is that the solution $S$ at the end of operation $i$ is {\em exactly} the output of \swapping on $\pi$. To see why this is the case, consider the story of each element $e$ arriving in $\pi$. There are two cases to consider. If $e$ is in some $X_{\ell}$, then our algorithm has added it to the candidate solution $S_{\ell}$ during the operation $i_{\ell}$ because $e$ was in either $E_{\ell}$ or $F_{\ell}$. Similarly, also \swapping would have added $e$ to its solution, with the exact same swap. If $e$ is in $Y_{\ell}$, then it means that the algorithm has filtered it out  during operation $i_{\ell}$ because it failed to meet the swapping condition: thus it would have been discarded also by \swapping. This implies, by \Cref{thm:swapping}, that $f(S)$ is a $4$-approximation to the best independent set in $X$, which is an upper bound on $\OPT$ (as it is the best independent set on a larger set).

    \begin{theorem}
    \label{thm:approx}
        For any operation $i$ it holds that the solution $S^i$ output by the algorithm at the end of the computation relative to iteration $i$ is a deterministic $4$-approximation of $\OPT_i$.
    \end{theorem}

\section{Running Time Analysis}

In this section, we analyze the amortized running time of our algorithm. Recall that, throughout this paper, we refer to running time as the total number of submodular function evaluation plus the number of independent set evaluation of the matroid. We start by showing some of the basic properties of the $A$ and $B$ sets.\\
\begin{observation} \label{obs:size-a-b}
    For any level $0 \leq \ell \leq L$, at the end of $\LevelConstruct(\ell)$, $|A_\ell| < \frac{n}{2^\ell}$ and $|B_\ell| = 0$. 
\end{observation}
\begin{proof}
    It follows directly from Line~\ref{line:level-until} in \LevelConstruct that $A_\ell < \frac{n}{2^\ell}$, otherwise the loop does not stop. Moreover, Line~\ref{line:level-b-empty} in \LevelConstruct is the only place that $B_\ell$ is changed and it is set to empty set.
\end{proof}

\begin{observation} \label{obs:size-b}
    For any level $0 \leq \ell \leq L$, during the execution of the algorithm $|B_\ell| \leq \frac{n}{2^\ell}$.
\end{observation}
\begin{proof}
    The only place where the size of $B_{\ell}$ increases is in Line~\ref{line:insertion-add-e} of \Insertion, where it increases by at most one. When $|B_\ell| = \frac{n}{2^\ell}$, then $\LevelConstruct(\ell)$ is called directly from Line~\ref{line:insertion-invokes-level-construct} of \Insertion or indirectly in Line~\ref{line:level-next-level} of \LevelConstruct. In both cases $|B_\ell| = 0$ due to \cref{obs:size-a-b}.
\end{proof}

    \begin{observation} \label{obs:size-a-always}
        For any level $0 \leq \ell \leq L$, during the execution of the algorithm $|A_\ell| \leq \frac{n}{2^{\ell-2}}$.
    \end{observation}
    \begin{proof}
        For any level $\ell$, the cardinality of $A_{\ell}$ only varies in two cases: when an element in $A_{\ell}$ is removed from the stream (line~\ref{line:deletion_candidate} of \Deletion) or during a call of $\LevelConstruct~$ on level $\ell$. Since the latter decreases the cardinality of $A_{\ell}$, we only study the former. When \LevelConstruct~$(\ell)$ is called, $A_{\ell}$ is initialized in line~\ref{line:level-first-set-a} and then its cardinality only decreases (Lines~\ref{line:A_l_filter}~and~\ref{line:level-sample}). To conclude the proof is then sufficient to prove that, every time $A_{\ell}$ is initialized in a new call of \LevelConstruct~$(\ell)$, its cardinality is at most $\tfrac{n}{2^{\ell-1}}$. The set $A_{\ell}$ is initialized with the elements in $B_{\ell-1}$ and in $A_{\ell-1}$. We know that $|B_{\ell-1}| \le \tfrac{n}{2^{\ell-1}}$ at any time of the algorithm (\Cref{obs:size-b}), while the cardinality of $|A_{\ell-1}|$ did not increase since the end of last call of \LevelConstruct~$(\ell-1)$. All in all, using the bound in \Cref{obs:size-a-b} we get the desired bound:
        \begin{align}\label{align:bound-a-all}
                |A_\ell| \leq \frac{n}{2^{\ell-1}} + \frac{n}{2^{\ell-1}} =  \frac{4n}{2^\ell}.
        \end{align}

        The size of $B_{\ell}$ only increases in Line~\ref{line:insertion-add-e} of \Insertion, where it increases by at most one. When $|B_\ell| = \frac{n}{2^\ell}$, then $\LevelConstruct(\ell)$ is called directly from Line~\ref{line:insertion-invokes-level-construct} of \Insertion or indirectly in Line~\ref{line:level-next-level} of \LevelConstruct. In both cases $|B_\ell| = 0$ due to \cref{obs:size-a-b}.
    \end{proof}

    Before moving to \LevelConstruct, we show that it is possible to compute the candidate swaps $s_e$ in line \ref{line:find-swap} in $O(\log k)$ calls of the independence oracle of the matroid.

\begin{lemma}
\label{lem:swap-speed-up}
    For any element $e \in A_{\ell}$ it is possible to find the candidate swap $s_e$ in line \ref{line:find-swap} of \LevelConstruct in $O(\log k)$ calls of the independence oracle of the matroid.
\end{lemma}
\begin{proof}
    Consider any iteration of the repeat loop in \LevelConstruct, let $S$ be the solution, $A$ the candidate set (we omit the dependence on $\ell$ for simplicity) and $e$ any element in it. If $S + e \in \cM$ then $s_e$ is set to the empty set and the claim holds. Otherwise, call $C$ the set of all elements in $S$ that can be swapped with $e$ to obtain an independent set:
    \[
        C = \{y \in S \mid S - y + e \in \cM\}.
    \]
    It is immediate to see that $s_e \in \argmin_{y \in C} w(y)$. We know that the solution $S=\{x_1,x_2,\dots, x_j\}$ is maintained in decreasing order of weights (resolving ties arbitrarily) and, by the downward closure property of matroids, we can use binary search to find 
    \[
        i^* = \max\{i \mid\{x_1,\dots,x_{i-1}\}+e \in \cM \}. 
    \]
    We claim that $x_{i^*}$ is a good choice of $s_e$, i.e., that $x_{i^*} \in \argmin_{y \in C} w(y)$.
    
    First, note that $x_
    {i^*}$ belongs to $C$. To see this, consider the set $R = \{x_1,\dots x_{i^*-1}\} + e \in \cM$, and recursively add element from $S$ to it while keeping $R$ independent. By the augmentation property, we know that this is possible until $|R| = |S|$. A single element remains in $S \setminus R$ and it has to be $x_{i^*}$, as we know that $\{x_1,\dots,x_{i^*}\} + e$ is dependent, thus $S - x_{i^*} + e = R \in \cM$. Now, we show that no element in $C$ can have smaller weight (i.e. larger index) than $x_{i^*}$. Assume toward contradiction that this is the case, i.e. that there is an $x_j$ such that $S - x_j + e \in \cM$ and $j < i^*.$ This implies that $\{x_1, \dots x_{j-1}\} + e$ is independent, which contradicts the minimality of $i^*.$
\end{proof}

\begin{lemma} \label{lemma:level-running-time}
    For any level $0 \leq \ell \leq L$, the running time of $\LevelConstruct(\ell)$ is $O( \frac{n\rank\log \Delta \log k}{2^{\ell}})$.
\end{lemma}
\begin{proof}
    
    We prove this Lemma in two steps. First, we control the running time of the non-recursive part of \LevelConstruct~$(\ell)$ (i.e., all the algorithm but the recursive call in line~\ref{line:level-next-level}), then we study, by induction, how these bounds combine recursively.

    We start proving that, for every level $\ell$ from $0$ to $L$, the running time of the non-recursive part of \LevelConstruct~$(\ell)$ is dominated by $c\frac{n\rank\log \rank\log \Delta}{2^{\ell}}$, for some positive constant $c$. Focus on any level $\ell$ and consider any call of \LevelConstruct~$(\ell)$. The main computation is performed in the repeat loop (Lines~\ref{line:level-from}-\ref{line:level-until}), which is repeated at most $\rank\log\Delta $ times (due to the exponential increase of the weight assigned to the elements, the fact we can at most add $k$ elements without swapping and the definition of $\Delta$). In each iteration of the repeat loop, the running time of finding the swap elements in Lines~\ref{line:begin-swap-loop}-\ref{line:level-set-a-loop} is at most order of $|A_{\ell}| \log \rank$, which is in $O( \frac{n\log \rank}{2^{\ell}})$ (recall, the cardinality of $A_{\ell}$ is bounded in \Cref{obs:size-a-always}). This concludes the analysis of the non recursive part of \LevelConstruct: there exists a constant $c>0$ such that its running time is at most $c\frac{n\rank\log \rank\log \Delta}{2^{\ell}}$, for every level $\ell$ from $0$ to $L$.

    We now conclude the proof of the Lemma by induction on $\ell$ from $L$ to $0$. More precisely, we show that the (overall) running time of $\LevelConstruct(\ell)$ is bounded by $2c\frac{n\rank\log \rank\log \Delta}{2^{\ell}}$ for any level $0 \leq \ell \leq L$. We start considering the base case $\ell = L$. As there is no recursive call, the inequality immediately descends on the bound on the non-recursive running time of \LevelConstruct~$(L)$. For the induction step, assume that the running time of $\LevelConstruct(\ell+1)$ is upper bounded by $2c\frac{n\rank\log \rank\log \Delta}{2^{\ell+1}}$ and we show that the same holds for level $\ell$. This is pretty easy to see: the interval running time of \LevelConstruct~$(\ell)$ is at most $c\frac{n\rank\log \rank\log \Delta}{2^{\ell}}$, while the recursive call to \LevelConstruct~$(\ell+1)$ has running time $2c\frac{n\rank\log \rank\log \Delta}{2^{\ell+1}}$ by the inductive hypothesis. Summing up these two terms yields the desired result.
\end{proof}

We have just assessed a deterministic upper bound on the computational complexity of each call of \LevelConstruct. We now bound the number of times that \LevelConstruct is called during the stream of insertions and deletions. To do so, we bound separately the number of times \LevelConstruct is directly induced by \Deletion or \Insertion.

\begin{lemma} \label{lemma:level-num-call-delete}
    For any level $0 \leq \ell \leq L$, the expected number of times that the {\LevelConstruct~$(\ell)$} function is called from \Deletion is at most $2^{\ell + 3}\rank \log \Delta$.
\end{lemma}
\begin{proof}
    Fix any level $\ell$. As a first step in this proof we show that the probability that a set of $\frac{n}{\rank 2^{\ell+1} \log \Delta}$ deletions hit at least one element sampled in Line~\ref{line:level-sample} of $\LevelConstruct(\ell)$ is at most $1/2$ (note, some of the elements sampled may get swapped out in the same $\LevelConstruct(\ell)$ execution). The reason is that there are at most  $\rank \log \Delta$ elements sampled from $A_{\ell}$. Each of these elements are sampled from a candidate pool of at least $\frac{n}{2^\ell}$ elements. Therefore the probability that any particular deleted element is from the sampled elements is at most $\frac{\rank 2^\ell \log \Delta}{n}$. The claim follows by union bound over all the $\frac{n}{\rank 2^{\ell+1} \log \Delta}$ deletions. 
    
    We call a period between two invocation of $\LevelConstruct(\ell)$ from $\Deletion$ an epoch, and denote with $N_{\ell}$ the (random) number of such epochs. We call an epoch short if its length is less than $\frac{n}{\rank 2^{\ell+1} \log \Delta}$ and long otherwise. We denote with $\shortN$ and $\longN$ the number of short, respectively long, epochs. Every time we recompute $\LevelConstruct(\ell)$, the probability of the next epoch being short is at most $1/2$. So we have:
    \begin{align}
    \label{eq:epochs}
        \E{\shortN} \leq \frac 12 \E{N_{\ell}} = \frac 12 \E{\shortN} + \frac 12 \E{\longN}, 
    \end{align}
    which implies that $\E{\shortN} \le \E{\longN}.$
    We know that the number $\longN$ of long epochs is at most the total number of operations $2n$, divided by the lower bound on the cardinality of each long epoch, $ \frac{n}{\rank 2^{\ell+1} \log \Delta}$. All in all, $\longN \le  2^{\ell+2} \rank \log \Delta$. We are ready to conclude:
    \[
         \E{N_{\ell}} =  \E{\shortN} + \E{\longN} \le 2 \E{\longN}\le  2^{\ell+3} \rank \log \Delta. \qedhere    
    \]
\end{proof}

\begin{lemma} \label{lemma:level-num-call-insert}
    For any level $0 \leq \ell \leq L$, the number of times that the $\LevelConstruct(\ell)$ function is called from \Insertion is at most $2^{\ell}$.
\end{lemma}
\begin{proof}
    The only place where the size of set $B_{\ell}$ increases is in  Line~\ref{line:insertion-add-e} of \Insertion, and it increases by at most one per insertion. Moreover, $B_\ell$ is set to the empty set in Line~\ref{line:level-b-empty} of $\LevelConstruct(\ell)$. Also there are at most $n$ insertions and $\LevelConstruct(\ell)$ is called when the size of $B_\ell$ is equal to $\frac{n}{2^\ell}$. Therefore there are at most $n\frac{2^\ell}{n} = 2^\ell$ calls to $\LevelConstruct(\ell)$ from \Insertion.
\end{proof}

\begin{lemma}\label{lem:running}
    The average running time per operation is $O(\rank^2\log \rank \log^2 \Delta \log n)$, in expectation over the randomness of the algorithm.
\end{lemma}
\begin{proof}
    The running time when an element is inserted or deleted is $O(L) =  O(\log n)$ beside the calls made to $\LevelConstruct~$. In what follows we focus on the total running time spent in $\LevelConstruct~$. There are two places that $\LevelConstruct(\ell)$ can be called from (beside the recursion in Line~\ref{line:level-next-level} of \LevelConstruct): \Insertion and \Deletion. By comparing \cref{lemma:level-num-call-insert} and \cref{lemma:level-num-call-delete} it results that the number of calls induced by \Deletion dominates those induced by \Insertion. So we only focus on the former term. Let $c$ be the constant as in the analysis of \cref{lemma:level-running-time}, we bound the total expected running time by 
    \begin{align*}
        2 \sum_{0 \leq \ell \leq L} \rank 2^{\ell+3} \log \Delta \cdot \left(  2c \cdot \frac{n\rank\log \rank\log \Delta}{2^{\ell}}\right)&= 32 c\sum_{0 \leq \ell \leq L} n\rank^2\log \rank \log^2 \Delta \\
        &= 32 c \cdot \left( n\rank^2 \log \rank \log^2 \Delta \log n\right). \qedhere
    \end{align*}
\end{proof}

\section{Putting it Together} \label{sec:all-together}
    The results in the previous sections hold under the assumption that the algorithm designer knows the number of insertion and deletions (denoted by $n$) in advance. In this section we present a well-known tool that enables our algorithm to run without this assumption. We simply start by $n=1$, and whenever the number of insertions and deletions reaches to $n$, we restart algorithm and double the value of $n$. Therefore, if the total operations is $m$, then largest value $n$ that we use is the smallest power of $2$ bigger than $m$, which is at most $2m$. Combining with \Cref{lem:running} we get that the total running time per operation is (up to multiplicative constants)
    \begin{align} \label{eq:all-runtime}
        \sum_{1 \leq n_0 \leq \log n} &\rank^2 \log \rank \log^2 \Delta n_0 = \rank^2 \log \rank \log^2 \Delta \log n^2 = \rank^2 \log \rank \log^2 \Delta \log^2 m\,.
    \end{align}

Combining \Cref{eq:all-runtime} and \Cref{thm:approx}, we have the main result of the paper.
\begin{theorem} \label{thm:main}
    Our algorithm yields a $4$-approximation to the fully dynamic monotone submodular maximization problem with matroid constraint and exhibit an $O(\rank^2\log \rank \log^2 \Delta \log^2 n)$ expected amortized running time.
\end{theorem}

In \cref{app:Delta} we also explain how one can avoid the dependency on $\Delta$ which requires: (i) designing and analyze a new algorithm that combines swapping and thresholding (presented in \cref{app:swapping}) and (ii) apply standard techniques to guess the value of  OPT and run multiple copies of the algorithm at the same time.

\begin{restatable}{corollary}{mainresultcoro}
    For any constant $\eps > 0$, there exists a $(4+O(\eps))$-approximation to the fully dynamic monotone submodular maximization problem with matroid constraint that exhibits an $O\left(\frac {\rank^2}\eps \log \rank \log^2 n \log^3 \frac \rank \eps\right)$ expected amortized running time.
\end{restatable}

\section{Conclusions and Future Directions} In this paper we design the first efficient algorithm for fully-dynamic submodular maximization with matroid constraint. An interesting open question stems immediately from our result: is it possible to reduce the amortized running to depend only poly-logarithmically in $k$ (currently it is $\tilde{O}(k^2)$)?

In this paper we focus on the crucial worst-case paradigm, constructing an algorithm whose guarantees are robust to any (oblivious) adversary that generates the stream of insertions and deletions. An interesting open problem of research is to study beyond-worst case analysis, when it is natural to assume some ``non-adversarial'' structure on the stream, similarly to what has been done, for instance, in the random-order arrival model for insertion-only streams \citep{Norouzi-FardTMZ18,LiuRVZ21,FeldmanLNSZ22}.

\section*{Acknowledgements}
The work of Federico Fusco is partially supported by ERC Advanced Grant 788893 AMDROMA “Algorithmic and Mechanism Design Research in Online Markets”, PNRR MUR project PE0000013-FAIR”, and PNRR MUR project  IR0000013-SoBigData.it. Part
of this work was done while Federico was an intern at Google
Research, hosted by Paul D\"utting.

\bibliographystyle{plainnat}
\bibliography{references}

\clearpage
\appendix

\section{\swapping algorithm}
\label{app:swapping}

    In this section we formally prove that the version of \swapping presented in \Cref{sec:prel} maintains the approximation guarantees of the original algorithm by \citet{ChakrabartiK15}. Actually, we prove a more general result that is instrumental to prove, in \Cref{app:Delta}, that it is possible to remove the dependence in $\Delta$ from the amortized running time. Consider the following version of \swapping, that we call \thresholdswapping, where all the elements with weight below a given threshold are simply ignored by the algorithm. The details are given in the pseudocode. Note that when the threshold is set to $0$, \thresholdswapping and  \swapping coincides. We prove now the following Theorem, that clearly implies \Cref{thm:swapping} by setting $\varepsilon$ and $\tau$ to $0$.

    \begin{algorithm}[H]
	\caption{\thresholdswapping} \label{alg:swapping_extended}
	\begin{algorithmic}[1]
		\STATE \textbf{Input:} rank $k$ of the matroid, precision parameter $\eps$ and threshold $\tau$
		\STATE \textbf{Environment:} stream $\pi$ of elements, function $f$, matroid $\cM$
		\STATE $S \gets \emptyset$, $S' \gets \emptyset$
		\FOR{each new arriving element $e$ from $\pi$}
		    \STATE $w(e) \gets f(e\mid S')$
		    \IF{$w(e) < \frac{\eps}{k} \tau$} \label{line:weight_test}
		        \STATE Ignore $e$ and \textbf{continue}
		    \ENDIF
		    \IF{$S + e \in \cM$}
		        \STATE $S \gets S + e$, $S' \gets S' + e$
		    \ELSE 
                \STATE $s_e \gets \argmin\{w(y) \mid y \in S, \ x + S - y \in \cM\}$
                \IF{$2 w(s_e) < w(e)$}
                    \STATE $S \gets S - s_e + e$, $S' \gets e + S'$
                \ENDIF
		    \ENDIF
		\ENDFOR
		\STATE \textbf{Return} $S$
	\end{algorithmic}
    \end{algorithm}
    
    \begin{theorem}
    \label{thm:tau_swapping}
        For any $\eps>0$ and threshold $\tau<\OPT$, it holds that \thresholdswapping yields a $(4 + O(\eps))$-approximation to the optimal offline solution on the stream.
    \end{theorem}

    The proof of \Cref{thm:tau_swapping} follows a similar analysis to the one of the original algorithm in \citet{ChakrabartiK15}; we precede it here with some Lemmata. 
    \begin{lemma}
    \label{cl:wproperties}
        Let $K=S'\setminus S$ be the set of elements that were in the solution and were later swapped out and extend by linearity the function $w$ to sets. Then the following three properties hold true:
    \begin{itemize}
        \item[$(i)$] $\ w(K) \le w(S)$
        \item[$(ii)$] $w(S) \le f(S)$
        \item[$(iii)$] $f(S') = w(S')$
    \end{itemize}
    \end{lemma}
    \begin{proof}
        Crucially, the weight function $w$ is linear and once an element enters $S$, its weight is fixed forever as the marginal value it contributed when entering $S$ but with respect to $S'$. During the run of the algorithm, every time an element $s_e$ is removed from $S$, the weight of $S$ increases by $w(e) - w(s_e)$ by its replacement with some element $e$. Moreover, $w(s_e) \le w(e) - w(s_e)$ for every element $s_e \in K$ since $2\cdot w(s_e) \leq w(e)$. Summing up over all elements in $K$, it holds that 
        \[
              w(K) = \sum_{s_e \in K}w(s_e) \leq \sum_{s_e \in K} \left[w(e) - w(s_e)\right] \le w(S),
        \]
        where the last inequality follows from a telescoping argument (for each chain of swaps we only retain with the positive sign the final elements that remained in $S$). This establishes $(i)$.

        Consider now the second inequality $(ii)$. We denote with $S'_e$ the version of set $S'$ when element $e$ was added, and we denote with $S_e = S'_e \cap S.$ We have the following:
        \[
            f(S) = \sum_{e \in S} f(e|S_e) \ge \sum_{e \in S} f(e|S'_e) = \sum_{e \in S} w(e) = w(S),
        \]
        where the crucial inequality is due to submodularity.
        Finally, the last point $(iii)$ follows by the definition of $w$.
    \end{proof}
    
    In the analysis of the approximation guarantees of \thresholdswapping we need to somehow account for all the elements in the optimum that were initially added to the solution by the algorithm but then later swapped out. To analyze this ``chain of swaps'' we need a useful combinatorial lemma: Lemma 13 of \citet{FeldmanK018} (also Lemma 30 of \citet{ChekuriGQ15}). It concerns directed acyclic graphs (DAGs) where the nodes are also elements of a matroid. Under some assumption, it guarantees the existence of an injective mapping between elements in an independent set and of the sinks of the DAG. As a convention, we denote with $\delta^+(u)$ the out-neighborhood of any node $u$ and we say that an independent set $T$ spans an element $x$ if $T + x \notin \cM$.  

    \begin{lemma}
    \label{lem:graph}
        Consider an arbitrary directed acyclic graph $G = (V, E)$ whose vertices are elements of some matroid $\cM$. If every non-sink vertex $u$ of $G$ is spanned by $\delta^+(u)$ in $\cM$, then for every set $S$ of vertices of $G$ which is independent in $\cM$ there must exist an injective function $\psi$ such that, for every vertex $u \in S$, $\psi(u)$ is a sink of $G$ which is reachable from $u$.
    \end{lemma}
    
    We clarify now how we intend to use the previous result to our problem, similarly to what is done in \citet{DuettingFLNZ22arXiv}. 
    \begin{lemma}
    \label{lem:graph_application}
        Let $\OPT$ be the optimal offline solution and denote with $F$ the set of elements that failed the threshold test in line \ref{line:weight_test} of \thresholdswapping. Then it holds that $w(\OPT \setminus (S' \cup F)) \le 2 w(S)$.
    \end{lemma}
    \begin{proof}
        Consider the following procedure to construct a DAG whose nodes are given by the elements of the stream that passed the weight test in line \ref{line:weight_test}. When an element $x$ arrives, call $C$ the circuit\footnote{A circuit is a dependent set that is minimal with respect to inclusion} in $S + x$ and $y$ the element in $C - x$ with smaller weight. If $y$ is swapped with $x$, then we add directed edges from $y$ to all the elements in $C - y$, if $x$ is not added, then add directed edges from $x$ to each element in $C - x$. If $x$ is added without any swap, then its out-neighborhood is empty. 
    Every edge in this graph points from a vertex dropped or swapped out at some time to a vertex that is either never deleted or removed at some time in the future. This time component makes the underlying graph a DAG.
    We now apply \Cref{lem:graph} on $G$, then there exists an injective function $\psi$ that associates each element in $\OPT \setminus (S'\cup F)$ to an element in $S$ such that $w(u) \le 2 \cdot w(\psi(u))$ for all $u \in \OPT \setminus (S'\cup F)$. 
    To see this, consider that in each $u$-$\psi(u)$ path there is at most one swapping where the weight does not increase, and it has to be the first swap if the new element $u$ of the stream was not added to the solution because of some node in the solution whose weight was possibly smaller, but no more than a factor $2$ smaller. This implies that $w(\OPT \setminus (S'\cup F)) \le 2 \cdot w(S).$
    \end{proof}
    
    We finally have all the ingredients to prove the Theorem.

    \begin{proof}[Proof of \Cref{thm:tau_swapping}]
        The proof of the Theorem is just a simple chain of inequalities that uses the Lemmata:
        \begin{align*}
            f(\OPT) &\le f(\OPT \cup S')\\
            &= f(S') + f(\OPT |S')\\
            &\le w(S) + w(K) + \sum_{e \in \OPT\setminus S'} f(e|S') \tag*{(Property $(iii)$ and submodularity)}\\
            &\le 2 \cdot f(S) + \sum_{e \in \OPT\setminus (S' \cup F)} w(e) + \sum_{e \in \OPT\cap F} w(e) \tag*{(Properties $(ii)+(i)$ and submodularity)}\\
            &\le 4 \cdot f(S) + \sum_{e \in \OPT\cap F} w(e) \tag*{(\Cref{lem:graph_application} and Property $(ii)$)}\\
            &\le 4 \cdot f(S) + \eps \OPT. \tag*{(Assumption on $\tau$)}
        \end{align*}
    \end{proof}

    \begin{algorithm}[H]
    	\caption{$\Insertion(e)$ corresponding to threshold $\tau = (1+\eps)^j$, for some integer $j$} \label{alg:insertion-fast}
    	\begin{algorithmic}[1]
    	    \IF{$(1+\eps)\tau  > f(e) \geq \frac{\eps}{k} \tau$} \label{line:insertion-range}
    		\STATE $B_\ell \gets B_\ell + e \quad \forall \, 0 \leq \ell \leq L$ 
    		\IF{there exists an index $\ell$ such that $|B_\ell| \ge \frac{n}{2^\ell}$}
    		    \STATE Let $\ellstar$ be such $\ell$ with lowest value 
        		\STATE Call $\LevelConstruct(\ellstar)$ 
    		\ENDIF
    		\ENDIF
    	\end{algorithmic}
    \end{algorithm}      
\section[Removing the dependence on Delta]{Removing the dependence on $\Delta$}
\label{app:Delta}

    In this section, we show how to remove the dependence in $\Delta$ of the amortized running time. Similarly to \citet{LattanziMNTZ20}, for any fixed choice of a precision parameter $\varepsilon>0$, we run multiple instances of our algorithm in parallel, where each instance is parametrized with a threshold $\tau = (1 + \eps)^j$ for various integer values of $j$. In each one of these instances, all the elements with weight smaller than $\varepsilon \cdot \tau/k$ are simply ignored. We return the best solution among these parallel runs after each operation. 
    
    There are two  modifications to the algorithm presented in the main body that we implement. One in the \Insertion routine, and one in \LevelConstruct. We present here the pseudocodes and describe the changes.

    The \Insertion routine is modified in a simple way: when an element is inserted, it is actually considered for insertion only if its value is within a certain range. We refer to the pseudocode for the details. Given this modification and the geometric construction of the thresholds $\tau =(1+\eps)^j$, we have an immediate bound on the number of parallel rounds of the algorithm that any element is actually inserted into:
    \begin{observation} \label{obs:num-copies}
        Any element $e$ is inserted in $O(\frac{1}{\eps} \log \frac k \eps )$ copies of the algorithm.
    \end{observation}

    We move our attention towards \LevelConstruct. There, we modify the repeat loop to filter out all the elements whose marginal contribution falls, at any point, below $\frac{\eps}{k} \tau$. More precisely, we add Line \ref{line:new-fast} in the following pseudocode. 
    
    \begin{algorithm}
	\caption{$\LevelConstruct(\ell)$ corresponding to threshold $\tau = (1+\eps)^j$, for some integer $j$} \label{alg:level-construct-fast}
	\begin{algorithmic}[1]
		\STATE $A_{\ell} \gets A_{\ell-1} \cup B_{\ell-1}$  \label{line:fastlevel-first-set-a}
		\STATE $B_\ell \gets \emptyset$ \label{line:fastlevel-b-empty}
		\STATE $S_{\ell} \gets S_{\ell-1}$ \label{line:fastlevel-s_l}
		\STATE $S'_{\ell} \gets S'_{\ell-1}$ \label{line:fastlevel-s'_l}
        \REPEAT	\label{line:fastlevel-from}
        \FOR{any element $e \in A_\ell$}\label{line:fastbegin-swap-loop}
            \STATE $w(e) \gets f(e \mid S'_{\ell})$ \label{line:fastlevel-weight}
            \STATE $s_e \gets \argmin\{w(y) \mid y \in S_{\ell} \wedge S_{\ell} - y + e \in \cM\}$\label{line:fastfind-swap}
        \ENDFOR
		\STATE $A_{\ell} \gets \{ e \in A_{\ell} \mid f(e \mid S_\ell) \geq \frac{\eps }{k} \tau \}$ \label{line:new-fast}
        \STATE $E_{\ell} \gets \{ e \in A_{\ell} \mid S_{\ell} + e \in \cM\}$ \label{line:fastE_l}
		\STATE $F_{\ell} \gets \{ e \in A_\ell \setminus E_{\ell} \mid w(e) > 2 \, w(s_e) \}$ \label{line:fastF_l}
		\STATE $A_{\ell} \gets E_{\ell} \cup F_{\ell}$ \label{line:fastA_l_filter}
        \label{line:fastlevel-set-a-loop}%
        \IF {$|A_{\ell}| \ge \frac{n}{2^\ell}$} \label{line:fastlevel-second-check}
        	\STATE Pop $e$ from $A_{\ell}$ uniformly at random \label{line:fastlevel-sample}
            \STATE $S_{\ell} \gets S_{\ell} + e - s_e$ 
            \label{line:fastadd_S_l}
            \STATE $S'_{\ell} \gets S'_{\ell} + e$\label{line:fastadd_S'_l}
        \ENDIF
        \UNTIL {$|A_{\ell}| < \frac{n}{2^\ell}$ } \label{line:fastlevel-until}
        \STATE \textbf{if} $\ell<L$, call $\LevelConstruct(\ell+1)$. \label{line:fastlevel-next-level}
    \end{algorithmic}
\end{algorithm}

    \mainresultcoro*
    
    \begin{proof}
        Let us start by the running time analysis. The crucial place of the analysis in the main body where $\Delta$ appears is when we use $\log \Delta$ to bound the length of any chain of swaps in any repeat loop of \LevelConstruct (\cref{lemma:level-running-time} and \cref{lemma:level-num-call-delete}). This is because each swap increase the weight by at least a factor $2$. 
        Let $n_j$ denote the number of elements inserted to the $j^{th}$ copy of the algorithm, corresponding to the threshold $\tau = (1+\eps)^j$. Our claim is that line~\ref{line:new-fast} enables us to bound the length of any chain of swaps by $\log \frac{k}{\eps}$. The argument is as follows: each element inserted into that copy of the algorithm has value (and, by submodularity also weight) at most $(1+\eps)\tau = (1+\eps)^{j+1}$. Conversely, only swaps with weight at least $\frac \eps k \tau = \frac \eps k (1+\eps)^j$ can happen (by the filter in line~\ref{line:new-fast} of \LevelConstruct). Putting these two bounds together, we have the desired upper bound on the length of any chain of swaps. 
        
        The previous argument allows us to replace $O(\log \Delta)$ with $O(\log \frac{k}{\eps})$ in the analysis of the running time of each copy of the algorithm (\cref{thm:main}): there the total running time is then $  O(n_j \cdot \rank^2\log \rank \log^2 \frac \rank \eps \log^2 n_j) $.
        Summing up over all the copies, we have that the total running time (up to multiplicative constants) is
        \begin{align*}
            \sum_{j:n_j>0} n_j \rank^2\log \rank \log^2 \frac \rank \eps \log^2 n_j &\le \left(\rank^2\log \rank \log^2 n \log^2 \frac \rank \eps \right)\sum_{j:n_j>0} n_j \\
            &\le \left( \rank^2\log \rank \log^2 n \log^3 \frac \rank \eps \right)\cdot \frac{n}{\eps}
        \end{align*}
        where in the first inequality inclusion we used the simple bound $n_i \le n$ and the second one follows from \cref{obs:num-copies}. Therefore the amortized running time is
        \[
            O\left(\frac {\rank^2}\eps \log \rank \log^2 n \log^3 \frac \rank \eps\right).
        \]
        It remains to show that our new algorithm does not significantly worsen the approximation algorithm of the algorithm we presented in the main body. Fix any operation $i$, and consider the copy of the algorithm corresponding to a threshold $\tau$ that lies in the interval $[\OPT_i/(1+\eps), \OPT_i]$. We first observe that any element $e \in V_i$ (any element currently part of the instance) cannot have $f(e) > (1+\eps) \tau$, since:
        \[
          f(e) \leq \OPT_i \leq(1+\eps) \tau\,.
        \]
        Therefore, ignoring element $e$ such that $(1+\eps)\tau  > f(e)$ does not affect the approximation guarantee. The elements whose value falls below $\frac \eps k \tau\ge \frac \eps k \OPT_i/(1+\eps)$ (and thus not inserted in this copy of the algorithm due to the filter in line \ref{line:insertion-range} of \Insertion), on the other hand, only cause an extra additive $O(\eps)$ error in the approximation and can be ignored (similarly to what has been used in the proof of \Cref{thm:tau_swapping}). Finally, from \cref{thm:tau_swapping} we know that the filter in line \ref{line:A_l_filter} of     \LevelConstruct only cause in another additive $O(\eps)$ error in the guarantee of \swapping. Using the same argument in \Cref{sec:appr} we can then conclude that our algorithm mutuates the approximation guarantees of \thresholdswapping and thus has a $4+O(\eps)$-approximation guarantee.
    \end{proof}

\section{\lazygreedy and \swapping fails in the dynamic setting}
\label{app:benchmark}

    In this section we show how two well known algorithms for submodular maximization subject to matroid constraints (\lazygreedy and \swapping) cannot be directly applied to the dynamic setting without suffering $\Omega(n)$ worst case update time. 
    
    We define the two fully-dynamic algorithms: \fullyswapping and \fullygreedy. They both maintain the set $V_i$ of elements that were inserted but not deleted in the first $i$ operations and a candidate solution. When an element from the candidate solution gets discarded they both recompute from scratch a feasible solution from $V_i$ using their non-fully-dynamic counterpart: \fullyswapping uses \swapping on some ordering of the elements in $V_i$, while \fullygreedy performs \lazygreedy on $V_i$. The two algorithms differ on how they handle insertions. When an element is inserted, \fullygreedy recomputes the solution from scratches on $V_i$ as the new element may have changed drastically the greedy structure of the instance. On the contrary, \fullyswapping simply processes this new element as \swapping would do; this is because \swapping is a streaming algorithm and thus handles efficiently insertions.

    We construct here an instance where both \fullyswapping and \fullygreedy exhibits an amortized running time that is $\Omega(n).$ Consider a set $V=\{x_1,\dots, x_n\}$ of $n$ elements and an additive function $f$ on it, with $f(x_i) = 3^i$. The stream is simple: the elements are inserted one after the other in increasing order $x_1, x_2, \dots$ and are then deleted one after the other in decreasing order $x_n, x_{n-1},\dots$ We show that both algorithms have a large per-operation running time on this instance. 

    \begin{claim}
        \fullyswapping has worst case amortized running time that is $\Omega(n)$
    \end{claim}
    \begin{proof}
        Consider the stream and the function presented above, with a cardinality constraint $k=1.$
        We divide the analysis of the running time into two parts. During the first $n$ operations, the algorithm performs at least one value call and one independence call. For each one of the deletion operations, the deleted element is --- by construction --- in the solution maintained by \fullyswapping. Let's call $x_i$ this deleted element, the routine \swapping is called on the sequence $x_1,\dots,x_{i-1}$, where it performs $\Omega(i)$ oracle calls. Summing up, the total number of oracle calls (both value and independence calls) is $\Omega(n^2)$, yielding an amortized running time of $\Omega(n)$. 
    \end{proof}
    
    \begin{claim}
        \fullygreedy has worst case amortized running time that is $\Omega(n)$
    \end{claim}
    \begin{proof}
        Consider the stream and the function presented above, with a cardinality constraint $k=1$.Also here, we divide the analysis of the running time into two parts. During the first $n$ operations, the algorithm performs at least one value call and one independence call (actually way more than that). For each one of the deletion operations, the deleted element is --- by construction --- in the solution maintained by \fullygreedy. Let's call $x_i$ this deleted element, the routine \lazygreedy is called on the sequence $x_1,\dots,x_i$, where it performs $\Omega(i)$ oracle calls. All in all, the total number of oracle calls (both value and independence calls) is $\Omega(n^2)$, yielding an amortized running time of $\Omega(n)$. 
    \end{proof}

\end{document}